\documentclass[a4paper,UKenglish,cleveref, autoref, thm-restate]{lipics-v2021}
\hideLIPIcs  %uncomment to remove references to LIPIcs series (logo, DOI, ...), e.g. when preparing a pre-final version to be uploaded to arXiv or another public repository

\title{PLS-completeness of string permutations} %TODO Please add

\author{Dominik Scheder}{TU Chemnitz, Germany}{dominik.scheder@informatik.tu-chemnitz.de}{0000-0002-9360-7957}{}%TODO mandatory, please use full name; only 1 author per \author macro; first two parameters are mandatory, other parameters can be empty. Please provide at least the name of the affiliation and the country. The full address is optional. Use additional curly braces to indicate the correct name splitting when the last name consists of multiple name parts.

\author{Johannes Tantow}{TU Chemnitz, Germany}{johannes.tantow@informatik.tu-chemnitz.de}{0009-0006-0408-6966}{}

\authorrunning{D. Scheder, J. Tantow} %TODO mandatory. First: Use abbreviated first/middle names. Second (only in severe cases): Use first author plus 'et al.'

\Copyright{Dominik Scheder, Johannes Tantow} %TODO mandatory, please use full first names. LIPIcs license is "CC-BY";  http://creativecommons.org/licenses/by/3.0/

%TODO mandatory: Please choose ACM 2012 classifications from https://dl.acm.org/ccs/ccs_flat.cfm 

 %TODO mandatory; please add comma-separated list of keywords

\category{} %optional, e.g. invited paper

\relatedversion{}
\relatedversiondetails{Full Version}{https://arxiv.org/abs/2505.02622}
\acknowledgements{We want to thank Neil Thapen for introducing the problem to us.}%optional

\nolinenumbers %uncomment to disable line numbering

\EventEditors{Anne Benoit, Haim Kaplan, Sebastian Wild, and Grzegorz Herman}
\EventNoEds{4}
\EventLongTitle{33rd Annual European Symposium on Algorithms (ESA 2025)}
\EventShortTitle{ESA 2025}
\EventAcronym{ESA}
\EventYear{2025}
\EventDate{September 15--17, 2025}
\EventLocation{Warsaw, Poland}
\EventLogo{}
\SeriesVolume{351}
\ArticleNo{54}
\usepackage{todonotes}
\usepackage{circuitikz}
\usepackage[caps, full]{complexity}

\newcommand{\ignore}[1]{{\color{blue}hidden stuff here; delete 
or remove ignore command}}

\newcommand{\x}{\mathbf{x}}
\newcommand{\y}{\mathbf{y}}
\newcommand{\e}{\mathbf{e}}
\newcommand{\globalmin}[1]{\textsc{#1-Permutation Global Orbit Minimum}}
\newcommand{\localmin}[1]{\textsc{#1-Permutation Local Orbit Minimum}}
\newcommand{\localminsolution}[1]{\textsc{Locally Minimal Solution}}
\newcommand{\group}[1]{\langle #1 \rangle}
\newcommand{\cost}{\textnormal{cost}}

\newcommand{\nth}[1]{#1\textsuperscript{th}}

\newcommand{\N}{\mathbb{N}}
\newclass{\CLS}{CLS}
\newlang{\FLIP}{FLIP}
\newlang{\DCR}{DCR}
\newlang{\sat}{SAT}

\newcommand{\Z}{\mathbb{Z}}

\begin{document}

\maketitle

\begin{abstract}
Bitstrings can be permuted via permutations and compared via the lexicographic order. In this paper we study the complexity of finding a minimum of a bitstring via given permutations. As finding a global optima is known to be \NP-complete\cite{DBLP:conf/stoc/BabaiL83}, we study the local optima via the class \PLS\cite{DBLP:journals/jcss/JohnsonPY88} and show hardness for \PLS. Additionally, we show that even for one permutation the global optimization problem is \NP-complete and give a formula that has these permutation as its symmetries. This answers an open question inspired from 
Kołodziejczyk and Thapen~\cite{DBLP:conf/sat/KolodziejczykT24} and stated at the \textit{SAT and interactions} seminar in Dagstuhl\cite{ThapenPersonal2024}.
\end{abstract}

\section{Introduction}
Given a string $\x \in \{0,1\}^n$ and a couple of permutations in $S_n$, 
we can apply a permutation to $\mathbf{x}$ and obtain a new bit string. What is the lexicographically 
smallest string we can obtain this way? This problem is known \cite{DBLP:conf/stoc/BabaiL83} to be \NP-hard. What about 
finding a local minimum, i.e., arriving at a bit string that cannot be further improved by 
a single application of a permutation? In this paper, we show that this problem is PLS-complete.\\
To be more formal, we are given a string $\x \in \{0,1\}^n$ and permutations $\pi_1, \dots, \pi_m$ on the set 
$[n]$. When we view $\x$ as a function $[n] \rightarrow \{0,1\}$, the notation $\x \circ \pi$ makes sense 
and is the string obtained from $\x$ by permuting the coordinates according to $\pi$. By 
$\group{\pi_{1},\dots,\pi_m}$ we denote the subgroup of $S_{[n]}$ generated by the $\pi_i$. The 
 problem \globalmin{$k$} asks for the $\pi \in \group{\pi_1,\dots,\pi_k}$ such that 
 $\x \circ \pi$ is lexicographically minimal. Babai and Luks \cite{DBLP:conf/stoc/BabaiL83} showed that this is \NP-hard even for $k=2$.
 In fact, we will see that it is \NP-hard even for $k=1$, i.e., a single permutation.\\

\localmin{$k$} asks for a local minimum. 
 That is, an element $\pi \in \group{\pi_1,\dots,\pi_k}$ such that 
 \begin{align*}
    \x \circ \pi
    \preceq_{\rm lex}
    \x \circ \pi \circ \pi_i 
    \quad \textnormal{ for all $1 \leq i \leq k$} \ , 
 \end{align*}
 i.e., a single application of a permutation $\pi_i$ cannot further improve the string $\x \circ \pi$.\\
A local optimum always exists and hence this is an instance of a total search problem. Total search problems where solutions are recognizable in polynomial time  form the class \TFNP.
Total search problems that can be stated as finding 
a local optimum with respect to a certain 
{\em cost function} and a {\em neighborhood relation}
constitute the subclass \PLS\
(polynomial local search).
Known  hard problems for \PLS\ include finding a pure Nash-equilibrium in a congestion game\cite{DBLP:conf/stoc/FabrikantPT04} or finding a locally optimal max cut (\textsc{LocalMaxCut}) \cite{DBLP:journals/siamcomp/SchafferY91}. Almost all known \PLS-complete problems require quite involved cost functions. Our problem \localmin{$k$} has the benefit 
of using the possibly simplest cost function - 
the lexicographic ordering. 
The only other \PLS-complete problem using 
a lexicographic cost function that we know of is 
\FLIP, which asks to minimize the $m$-bit
output of a circuit $C$, where the solutions 
are all $n$-bit inputs and the neighborhood relation
is defined by flipping a single bit. 

Thus, our result unifies two desirable properties - 
our \PLS-complete problem is very combinatorial 
in nature (in contrast to \FLIP) and uses 
a very simple cost function (in contrast to 
\textsc{LocalMaxCut})

\subsection{SAT Solving and Symmetry Breaking}

When encoding a combinatorial problem as a CNF formula $F$ 
(think of ``is there a $k$-Ramsey graph on $n$ vertices?''), 
the formula will often contain many symmetries. To make the problem
easier for SAT solvers, one can take the statement
\begin{quotation}
\emph{
The satisfying assignment $\alpha$ should be a local lexicographical minimum
with respect to those symmetries,}
\end{quotation}
encode it as a CNF formula $G$ and feed $F \wedge G$ to 
the SAT solver. Clearly, $F \wedge G$  is satisfiable 
if and only if $F$ is.
In case that $F \wedge G$ is unsatisfiable, 
SAT solvers are often expected to produce a proof of 
unsatisfiability. A popular proof system used in this context is DRAT~\cite{DBLP:conf/sat/WetzlerHH14}. 
However, it is not known to what extend DRAT can handle 
symmetry breaking~\cite{DBLP:conf/cade/HeuleHW15}, that is, 
whether a short DRAT-refutation of $F \wedge G$ can be transformed
into a short DRAT-refutation of $F$. In this context, 
Thapen~\cite{ThapenPersonal2024} asked whether there exists 
a polynomial algorithm that, given a CNF formula $F$, 
a handful of symmetries thereof, and a satisfying assignment $\alpha$,
finds a satisfying assignment $\beta$ that is a local 
lexicographical minimum with respect to those symmetries. 
In this paper, we show that this problem is \PLS-complete, which 
is evidence that such a polynomial time algorithm might not exist.

\section{Preliminaries}

\subsection{Total search problems and \PLS}

The class \FNP\ is the functional correspondent of \NP.
\TFNP\ $\subset$ \FNP\ is the subset of {\em total} 
search problems, i.e., problems that always have a solution. As this is a semantic class, \TFNP{} has no known complete problems, and thus it is usually studied via its subclasses. These subclasses are based on the combinatorial principle that proves the existence of a solution. These principles include the existence of sinks in directed acyclic graphs (\PLS)\cite{DBLP:journals/jcss/JohnsonPY88}, the parity argument for directed and undirected graphs (\PPAD, \PPA) or the pigeonhole principle (\PPP) (all introduced in \cite{DBLP:conf/focs/Papadimitriou90}). Nonetheless, not all problems in \TFNP\ can be categorized in one of the known subclasses, \textsc{Factoring} being a prime example.

By the above characterization, \PLS\ requires finding a sink of a directed acyclic graph $G$. This would be possible in polynomial time if $G$ was given explicitly. Instead, $G$ is always given implicitly via a circuit that computes the successor list of a given node. To make sure that $G$ is acyclic, 
we have a second circuit computing a topological ordering, that 
is, a ``cost'' function that is strictly decreasing along 
the edges. A solution to the problem is a sink of $G$ 
or an edge $(u,v)$ with $\cost(u) \leq \cost(v)$, i.e., 
violating the decreasing cost condition. This 
guarantees the totality of the problem.  

An alternative definition is the following. For a \PLS\ problem $P$ we have a set of instances $I$. Each instance $i \in I$ has a set of feasible solutions $S$ (e.g. for the Euclidean traveling salesman problem the solutions are exactly 
the Hamilton cycles of $K_n$). Additionally, we require the following polynomial-time computable algorithms (usually given as circuits):
\begin{enumerate}
    \item an algorithm that decides whether a given $s$ 
    is a feasible solution.
    \item an algorithm that computes a starting solution $s \in S$
    \item an algorithm that computes for a feasible solution $s$ the neighborhood $N(s)$
    \item an algorithm that computes for a solution $s$ the cost $\cost(s)$ 
\end{enumerate}

A feasible solution $s$ is a {\em local minimum} if 
$\cost(s) \le \cost (s')$ for all $s' \in N(s)$. The definition is given for a minimization problem, but can be also defined in terms of a maximization problem. 

We can easily transform this into a directed acyclic graph by keeping only those neighbors in $N(s)$ having strictly 
smaller cost--or even keeping only the one neighbor 
$s' \in N(s)$ of minimal cost (breaking ties arbitrarily).

Similar to \NP, there is also a hardness structure in \PLS. Let $P$ and $Q$ be two problems in \PLS. $P$ reduces to $Q$ via a \emph{\PLS-reduction} $(f,g)$ for functions $f$ and $g$ such that $f$ maps an instance $I$ from $P$ to an instance $f(I)$ of $Q$ and $g$ maps a solution $s$ of $f(I)$ to a solution $g(I, s)$ of $P$ so that if $s$ is a local minimum in $f(I)$ then also $g(I,s)$ is a local minimum in $I$. This was defined in \cite{DBLP:journals/jcss/JohnsonPY88} and the first natural \PLS-complete problem is \FLIP. There solutions are $n$-bit strings, the cost is calculated by a given circuit $C$ and the neighborhood are all $n$-bit strings with a Hamming distance of 1.

The obvious greedy algorithm to find a solution for a \PLS-problem is as follows: Use the given algorithms to compute the start solution and always select the best neighbor until there is no better solution. This solution is called the \emph{standard solution} and the algorithm the \emph{standard algorithm}. For the problem \FLIP, finding the standard solution for a given start solution is \PSPACE-complete\cite[Lemma 4]{DBLP:conf/stoc/PapadimitriouSY90}.

Reductions that preserve the \PSPACE-completeness are called tight\cite{DBLP:journals/siamcomp/SchafferY91}. For this, we consider the transition graph $TG(I)$ of an instance $I$ of the problem $P$, that has a directed edge from each feasible solution $x$ to all of its neighbors $N(x)$.

A \PLS-reduction $(f,g)$ from $P$ to $Q$ is called \emph{tight} if for every instance $I$ of $P$ there exists a set $\mathcal{R}$ of feasible solutions for $f(I)$ such that
\begin{enumerate}
    \item $\mathcal{R}$ contains all local optima of $f(I)$
    \item For every solution $s$ of $I$ it is possible to construct in polynomial time a feasible solution $t \in \mathcal{R}$ such that $g(I,t) = s$
    \item If the transition graph of $TG(f(I))$ contains a path from $q$ to $q'$ such that both $q$ and $q'$ are in $\mathcal{R}$ and all other intermediate nodes are not in $\mathcal{R}$, let $p = g(I, q)$ and $p' = g(I, q')$ be the corresponding solutions in $P$. Then either $p = p'$ or there is an arc from $p$ to $p'$ in $TG(I)$.
\end{enumerate}

An interesting subclass of \PLS\ is \CLS\ that is supposed to capture continuous local search problems. Recently, it was shown that \CLS\ $=$ \PLS\ $\cap$ \PPAD\cite{DBLP:conf/stoc/FearnleyGHS21}.

\subsection{Permutation groups}

A permutation of a set $V$ is a bijection $\pi: V \rightarrow V$. 
For permutations $\pi_1, \dots, \pi_k$, we denote by 
$\group{\pi_1, \dots, \pi_k}$ the subgroup of $S_V$ generated
by the $\pi_i$. Checking membership
$\pi \in \group{\pi_1,\dots,\pi_k}$ is non-trivial 
but can be done in polynomial time~\cite[Section 1]{DBLP:conf/focs/FurstHL80}. 

\section{Related Work}

Many problems are known to be \PLS-complete, whose reductions
mostly start from \FLIP.  An influential reduction technique was used by Krentel~\cite{Krentel89} to show that finding a local minimum of a weighted CNF-formula is \PLS-complete. The idea is to have multiple copies of the circuit, so called \emph{test circuits} that precompute the effects of a flip. 
This idea is pushed further in \cite{DBLP:journals/siamcomp/SchafferY91} in order to show that finding a local minimum for weighted positive NAE(not all equal) 3-SAT is \PLS-complete, and this is used to show that other problems as finding a \textsc{LocalMaxCut} or finding a stable configurations of Hopfield networks are \PLS-complete.

Other direct reductions from \FLIP\ are used in \cite{DBLP:conf/stoc/FabrikantPT04} to show that finding a pure Nash equilibrium in an asymmetric network congestion game is \PLS-complete and in \cite{DBLP:journals/tcs/DumraufM13} to show that maximum constraint assignment, a generalization of CNF-SAT, is \PLS-complete. These reductions are especially of interest to us as \FLIP\ is to the best of our knowledge the only problem in \PLS\ with lexicographical weights and hence needed for our reduction.

Numerous optimization problems on permutation groups given via generators $\pi_1, \dots ,\pi_k$ are studied in \cite{DBLP:journals/disopt/BuchheimJ05}.
These include finding a $\pi \in \group{\pi_1, \dots, \pi_k}$ that minimizes $\Sigma_{i \in V} c(i, \pi(i))$ for some cost function $c: V^2 \to \mathbb{R}$. All these problems are 
shown to be \NP-complete via a reduction from finding a fixed-point free permutation, which is shown to be \NP-complete.

Our problem has previously been studied in \cite{DBLP:conf/stoc/BabaiL83}, but there the interest was in {\em global} optima. This was proven to be \NP-hard even for abelian groups.

\section{The one permutation problem}

We start our investigation by studying the special case that $k=1$, i.e., 
we have only one generator $\pi_1$ and our subgroup 
is $\group{\pi_1}$. In this case one can efficiently find a local optimum\cite[Section 8.2]{DBLP:conf/sat/KolodziejczykT24}. We show this with a different algorithm again. In contrast, we will show that surprisingly finding 
a global optimum is \NP-complete. To the best of our knowledge, this has only been known for {\em two} permutations~\cite{DBLP:conf/stoc/BabaiL83}.

\subsection{Finding a local optimum}

The problem is efficiently solvable when $k=1$, i.e., 
we are only given a single permutation $\pi$. We describe a different way of solving it in contrast to \cite{DBLP:conf/sat/KolodziejczykT24}. We can transform $\pi$ into the cycle notation and annotate each element with the bit that it is mapped to by the string $\x \in \{0,1\}^n$. We call a cycle {\em interesting} if $\x$ is non-constant on it. Additionally, we identify each cycle with its smallest member.

Consider the permutation $(1\ 2\ 5) (3\ 4) (7\ 8)$ and the string $001001$ depicted in \cref{ex:onePerm}. We color an element green if it is mapped to zero and orange if it is mapped to one by the string.
The left cycle is the cycle of $1$ and not interesting whereas the other two are which are identified by $3$ and $7$.

\vspace{.5cm}

\begin{figure}
\centering
\begin{tikzpicture}
    \node[circle, draw, fill=green] (one) at (0,0) {1};
    \node[circle, draw, fill=green] (two) at (2,0) {2};
    \node[circle, draw, fill=green] (five) at (1,1) {5};
    \node[circle, draw, fill=orange] (three) at (4,0) {3};
    \node[circle, draw, fill=green] (four) at (4,1) {4};
    \node[circle, draw, fill=green] (seven) at (6,0) {7};
    \node[circle, draw, fill=orange] (eight) at (6,1) {8};

    \draw[->] (one) -- (two);
    \draw[->] (two) -- (five);
    \draw[->] (five) -- (one);
    \draw[->] (three) -- (four);
    \draw[->] (four) -- (three);
    \draw[->] (seven) -- (eight);
    \draw[->] (eight) -- (seven);

\end{tikzpicture}
\caption{An example of the scenario for one permutation $\pi$}
\label{ex:onePerm}
\end{figure}

The permutation has no effect on elements on non-interesting cycles. Due to the cost function, lower indices are more costly than higher indices. We look for the interesting cycle that contains the position with the least index among all interesting cycles and 
let $l$ be the smallest index on it. 
There are indices
$i$ and $j := \pi(i)$ on this cycle with $x_i=0$ and $x_j=1$.
Let $k$ be such that $\pi^k(l) = i$. Then 
$\x \circ \pi^k$ has a $0$ at position $i$ but 
$\x \circ \pi^{k+1}$ has a $1$. In other words,
$\x \circ \pi^k$ is a local optimum

In example in \cref{ex:onePerm} we have $x=3$, $j=4$ and $d=1$. The local optimal permutation is hence $\pi$.

\subsection{Finding a global optima}

\begin{theorem}[\cite{DBLP:conf/stoc/BabaiL83}]
\globalmin{$2$} is \NP-hard.
\end{theorem}
This follows via a reduction from 
Independent Set where we encode a graph as a bit 
string, one bit per potential edge, 
and the permutations basically allow 
us to move the vertices of the independent 
set $I$ to the front, generating a large prefix 
of ${|I| \choose 2}$ many $0$'s.

Since \localmin{$1$} was so clearly solvable in polynomial time (even by the greedy algorithm), 
it comes as a surprise that global optimization, even for {\em one} permutation, is \NP-hard:

\begin{theorem}
\label{theorem-global-1-hard}
\globalmin{$1$} is \NP-hard.
\end{theorem}

We will define an intermediate \NP-complete problem called \textsc{Disjunctive Chinese Remainder}, 
short \DCR. For two numbers $t, m \in \N$ and a set $S \subseteq \N$, we write 
\begin{align*}
     t \not \in S \mod m 
\end{align*}
to state that $t \not \equiv s \mod m$ for all $s \in S$. Now in 
the \DCR\ decision problem, we are given moduli $m_1, \dots, m_l$ and sets of ``forbidden remainders''
$S_1, \dots, S_l$ with $S_i \subseteq \Z_{m_i} := \{0, 1, \dots, m_i - 1\}$. All numbers are given in unary and 
the moduli are not required to be pairwise co-prime. 
\DCR\  asks for a solution $t \in \N$ of  the system
\begin{align*}
   t & \not \in S_i \mod m_i \qquad \forall i = 1, \dots, l \ . 
\end{align*} 
This is clearly in \NP: if there is a solution $x \in \N$, then there is one with 
$0 \leq t \leq {\rm lcm}(m_1, m_2, \dots m_l)$ and thus $t$ has polynomially many bits in binary. 
Verifying that this $t$ is a solution can now be done by division with remainder.

\begin{lemma}
\label{DCR-is-hard}
\DCR\ is \NP-complete.
\end{lemma}
\begin{proof}
 We reduce from $3$-Colorability. Given a graph $G = (V,E)$ with $|V| = n$, we 
 let $3 = p_1, p_2, \dots, p_n$ be the first $n$ prime numbers greater than 2. By the prime number theorem,
 $p_n$ is polynomial in $n$, thus the $p_i$ can be found in polynomial time by 
 a brute force search using naive prime number testing.\\
 
 We let $N = p_1 \cdot p_2 \cdot \dots \cdot p_n$ and define the following function 
 from $\Z_N$ to $3$-colorings of the vertices $V$: given a number $0 \leq t \leq N-1$, 
 the corresponding coloring $c_t : V \rightarrow \{r,g,b\}$ is defined by 
 \begin{align*}
    c_t(v_i) & = \begin{cases}
        r & \textnormal{ if $t \equiv 0 \mod p_i$} \\
        g & \textnormal{ if $t \equiv 1 \mod p_i$} \\
        b & \textnormal{ else.} 
    \end{cases}
 \end{align*}
 By the Chinese Remainder Theorem, this is a surjective function and thus every 
 3-coloring can be encoded by one single number $0 \leq x \leq N-1$. For an 
 edge $e = \{u,v\}$, we write a constraint that makes sure that $u$ and $v$ receive different 
 colors. For each pair $(a,b) \in \Z_{p_u} \times \Z_{p_v}$ with 
  $(a,b) = (0,0)$ or $(a,b)=(1,1)$ or $a \geq 2, b \geq 2$, we compute the
 unique number $c \in \Z_{p_uq_v}$ with $c \equiv a \mod p_u$ and $c \equiv b \mod q_v$. 
 Let $S_{e}$ be the set of all numbers $c$ thus constructed, set $m_e := p_u \cdot p_v$ and 
 write the constraint 
 \begin{align}
    x \not \in S_{e} \mod m_e \ .  \label{constraint-modulo-prime-product}
 \end{align}
 If $e$ is the \nth{$i$} edge of the graph, we set $m_i = p_u p_v$ and $S_i = S_{e}$.
 We see that $x$ satisfies (\ref{constraint-modulo-prime-product}) if and only if 
 $x$ encodes a properly colors edge $e$.
\end{proof}

\begin{lemma}
\label{DCR-to-Min}
\DCR\ reduces to \globalmin{$1$}.
\end{lemma}
\begin{proof}
  Given an instance of \DCR\, we set $M = m_1 + m_2 + \cdots + m_l$ and 
  define a permutation $\pi$ on $[M]$ that has $l$ disjoint cycles, one of length $m_i$ for each $i$. 
  We label the elements of the \nth{$i$} cycle consecutively with the numbers $0,\dots,m_i-1$. 
  We define a bit string $\x \in \{0,1\}^M$ that has exactly one $1$ in each cycle, 
  placed on the element that has label $0$. The orbit element $\x \circ \pi^t$ still 
  has exactly one $1$ in cycle $i$, but now at the vertex labeled $t \textnormal{ mod } m_i$.

  For each cycle $i$, let $F_i$ be the elements on it whose labels are in the set $S_i$ of forbidden
  remainders and let $F := F_1 \cup \dots \cup F_l$. 
  We order the elements of $[M]$ such that the elements of $F$ come first.
  There exists a solution $t \in \N$ to the \DCR\ instance if and only if 
  the string $\x \circ \pi^t$ has no $1$ in any position in $F$. 
\end{proof}

\section{\PLS-Hardness}

We now state and prove our main result:
\begin{theorem}
  \localmin{$k$} is \PLS-complete.
  \label{theorem-PLS-complete}
\end{theorem}

In order to view it as a \PLS-problem we have an instance $I$ consisting of the permutations $\pi_1, .., \pi_k$ and the string $s \in \{0,1\}^N$. Solutions are permutations from $\langle\pi_1, \dots \pi_k\rangle$ as we can efficiently recognize whether permutations are in the group generated by $\pi_1, .., \pi_k$. 
The start solution is the identity. Neighbors of $\pi$ are permutations $\pi \circ \pi_i$ for $i \in \{1, \dots k\}$. The cost of a permutation $\pi$ is $\sum_{i=1}^{N} s(\pi(i)) 2^{N-i}$, where $s(\pi(i))$ is the digit of the position that $i$ is mapped to by $\pi$. All these can be computed in polynomial time, hence the problem is in \PLS.

A solution $\pi$ is cheaper than a solution $\sigma$ for a string $s$ if there is an integer $i$ such that for all $j < i$ we have that $s(\sigma(j)) = s (\pi(j))$ and $s(\pi(i)) < s(\sigma(i))$ as the cost is a geometric sum.

We can turn the minimization problem into a maximization problem by inverting the string.

\subsection{High-level idea}
We reduce from the \PLS-complete problem \FLIP. This problem is especially suitable since it uses a lexicographic cost function, too. Formally \FLIP\ is defined as follows:

\begin{definition}
    An instance of \FLIP\ consists of a circuit $C$ with $n$ inputs and $m$ outputs. Feasible solutions are all input assignments, i.e., the set $\{0,1\}^n$. The cost of a solution $\x$ is the output of $C(\x) \in \{0,1\}^m$. Two solutions are neighbors if they differ in a single bit.
    The cost function is defined by reading the output 
    as a number in binary; in other words, by 
    the lexicographic order on $\{0,1\}^m$. We are asked to find a solution whose cost is minimal among all its neighbors.
\end{definition}

We use an idea by Krentel\cite{Krentel89} and have $n+1$ copies $C_0, C_1, \dots, C_n$ of the circuit $C$, where $C_0$ 
is fed as an input $\x$ and 
$C_i$ is fed as an input $\x \oplus \e_i$, i.e., $\x$
with the $i$-th bit flipped. The setup is depicted in \cref{fig:highlevel}. 

The permutation group consists of two types of permutations. The first kind  $\pi_i^j$ simulates flipping the output of gate $i$ in circuit $j$; we allow for gates and hence circuits to be temporarily evaluated incorrectly. The input and output of a gate are not saved anywhere but syntactically built into the permutations and string. An exception is the input and output of the circuit, which we save for later usage. We use some positions as very important control bits to ensure that the correct evaluation is always possible.

The second type of permutation $\sigma_j$ swaps the circuit $0$ with the circuit $j$ and flips the $j$-th input bit for all 
other circuits. This simulates a step in the FLIP-problem. 

\begin{figure}[h]
    \centering
    \begin{tikzpicture}[
    box/.style={draw, thick, minimum width=5cm, minimum height=1cm},
    smallbox/.style={draw, thick, minimum width=0.6cm, minimum height=0.6cm},
    bit/.style={draw, rectangle, minimum width=0.6cm, minimum height=0.6cm, thick},
    swapline/.style={thick, rounded corners, <->}
    ]

    % First circuit
    \node[bit] (input1a) at (0, 0.3) {0};
    \node[bit] (input1b) at (0, -0.3) {1};
    \node[bit] (input1c) at (0, -0.9) {1};     
    \node[box] (C0) at (4, 0) {$C_0$};
    \node[bit] (output1a) at (8, 0.3) {0};
    \node[bit] (output1b) at (8, -0.3) {1};   
    \node[smallbox] (gi) at (5,0) {$g_i$};
    \node[smallbox] (gik) at (6,0) {$g_{i+k}$};
    \path[<->] (gi) edge [loop above] node {$\pi_i^0$} (gi);
    \path[->] (gik) edge [loop above] node {$\pi_i^0$} (gik);
    
    % Second circuit
    \node[bit] (input2a) at (0, -2) {1};
    \node[bit] (input2b) at (0, -2.6) {1};
    \node[bit] (input2c) at (0, -3.2) {1};  
    \node[box] (C1) at (4, -2.3) {$C_1$};
    \node[bit] (output2a) at (8, -2) {0};
    \node[bit] (output2b) at (8, -2.6) {0};       
    % Third circuit
    \node[bit] (input3a) at (0, -4) {0};
    \node[bit] (input3b) at (0, -4.6) {0};
    \node[bit] (input3c) at (0, -5.2) {1};  
    \node[box] (C2) at (4, -4.3) {$C_2$};
    \node[bit] (output3a) at (8, -4) {1};
    \node[bit] (output3b) at (8, -4.6) {1};   

    % Fourth circuit
    \node[bit] (input4a) at (0, -6) {0};
    \node[bit] (input4b) at (0, -6.6) {1};
    \node[bit] (input4c) at (0, -7.2) {0};  
    \node[box] (C3) at (4, -6.3) {$C_3$};
    \node[bit] (output3a) at (8, -6) {1};
    \node[bit] (output3b) at (8, -6.6) {0};

    \draw[swapline] (C1.west |- input1b.east) -- ++(-1, 0) -- ++(0, -1) -- node[midway, left] {$\sigma_1$}(C2.west |- input2a.east);

    \path[->] (input3a) edge [loop left] node {$\sigma_1$} (input3a);
    \path[->] (input4a) edge [loop left] node {$\sigma_1$} (input4a);

    \end{tikzpicture}
    \caption{A high level overview of the reduction}
    \label{fig:highlevel}
\end{figure}

\subsection{Definition of the reduction}

We assume without loss of generality that the circuit $C$ consists only of NAND-gates. Additionally, we assume that no input is directly passed to an output, i.e. on every path from an input to an output, there is at least one gate. 

We will now describe the set $V$ of {\em positions}, 
the permutations in $S_V$, and how strings in $\{0,1\}^V$, 
called {\em assignments}, 
correspond to the circuits $C_0,\dots, C_n$ computing 
their values on an input $\x$ and its neighbors 
$\x \oplus \e_i$. We face two main challenges:
\begin{enumerate}
 \item We  need an operation of the form
 ``flip bit $i$'', but permutations can only permute 
 positions, not flip bits. We solve this 
 by replacing position $i$ by two positions 
 $i_0, i_1$ and encoding $[y_i = 0]$ by 
 $[y_{i_0}=0, y_{i_1}=1]$ and $[y_i=1]$ by
 $[y_{i_0}=1, y_{i_1}=0]$.
 Flipping bit $i$ then corresponds to the 
 permutation that swaps $i_0$ and $i_1$, i.e., the 
 transposition $(i_0, i_1)$. 
 We call the one-position-per-bit view the 
 {\em condensed view} and the two-positions-per-bit view 
 the {\em expanded view}. We will give details 
 in the full version due to space restrictions. For now, we phrase things in the condensed view.
 \item Usually, when we flip an input $x_i$ to a circuit $C$, we 
 imagine the change propagating instantaneously through 
 the circuit $C$, potentially changing its output. 
 Here, we need to allow for a way for this change to 
 proceed gradually; therefore, we allow gates 
 to be temporarily in an incorrect state.
\end{enumerate}

When we have a position $i \in V$ and 
an assignment $\y \in \{0,1\}^V$ and $y_i = b$, 
we sometimes say {\em $i$ is assigned value $b$}
and sometimes {\em the label of $i$ is $b$}.\\ 

\textbf{State of a gate.} The {\em state of a gate}
is a triple $(x,y,b)$ where $x, y$ are the 
two input bits and $b$ is the output bit. 
If $b = \neg (x \wedge y)$ we call $(x,y,b)$ {\em correct}, 
because $b$ is what it is supposed to be: the NAND of 
$x$ and $y$; otherwise, we call it {\em incorrect}.
A gate $g$ is represented by a gadget of
four positions $g_{00}, g_{01}, g_{10}, g_{11}$, 
ordered as a $2\times 2$ square,
as shown in Figure~\ref{fig:enter-label}. 
A state $(x, y, b)$ is encoded as follows: 
position $g_{xy}$ is labeled $b$, the three
remaining ones are labeled with $\neg b$. 
Note that not all labelings of the gadget correspond 
to a gate state, only those where the number of positions 
labeled $1$s is one or three.
\begin{observation}
   The triple $(x,y,b)$ is correct if and only 
   if the position $g_{11}$ in its representation 
   is labeled $0$.
   \label{observation-control-bit}
\end{observation}
This is the core reason why we use this representation: 
we can determine correctness of a gate by reading 
just one bit. This will be important later when 
defining a cost function: having a $1$ at those 
{\em control positions} is bad. \\

\textbf{Input and output variables.}
Let $x_i^{(j)}$ be the $i$-th input variable 
to the $j$-th circuit (so $1 \leq i \leq n$ and 
$0 \leq j \leq n$). We introduce one position for 
each $x_i^{(j)}$. Similarly, let $c_k^{(j)}$ be the 
$k$-th output value of the $j$-th circuit; we introduce 
one position for each $c_k^{(j)}$. 

A central definition is that of a well-behaved 
assignment. Basically, it formalizes when an assignment 
encodes the partial evaluation of the inputs by the 
$n+1$ circuits. 

\begin{definition}
\label{definition-well-behaved}
 An assignment $\y \in \{0,1\}^V$ is called 
 {\em well-behaved} if the following hold:
 \begin{itemize}
 \item For each gate $g$, the four positions 
 $g_{00}, g_{01}, g_{10}, g_{11}$ are the 
 encoding of a gate state $(a_1, a_2, b) \in \{0,1\}^3$. 
 \item If the output of a gate $g$ is the $l$-th 
 input of some gate $h$, then the corresponding input 
 and output values agree, 
 i.e., if the state of $h$ is $(a_1', a_2', b')$ then 
 $b = a_l'$.
 \item If the $l$-th input of $g$ is an input variable 
 $x_i^{(j)}$, then $x_i^{(j)}$ is assigned the value
 $a_l$. 
 \item If $g$ is the $k$-th output gate of circuit 
 $j$, then its output value $b$ is the same as  the label 
 of $c_k^{(j)}$.
 \item The labels of the input values $x_{i}^{(j)}$ equal 
 $x_i^{(0)}$ if $i \ne j$, and are unequal if $i = j$; 
 in words, if the labels of the input positions of $C_0$
 form a vector $\x \in \{0,1\}^n$, then those of $C_j$ form 
 the vector $\x \oplus e_j$. 
 \end{itemize}
\end{definition}

Next, we describe the permutations on the positions. 
They come in two types: (1) flipping a gate and 
(2) swapping two circuits. \\

\textbf{Flipping a gate.} If $g$ is in state 
$(a_1, a_2, b)$, then flipping $g$ means replacing 
$b$ by $\neg b$, and for each gate $h$ (in state 
$(a'_1, a'_2, b')$) into which 
$g$ feeds as $l$-th input, flipping $a'_l$. 
Since our permutations do not work on the {\em state} 
of a gate but on its representation in the 
four-position gadget (Figure~\ref{fig:enter-label}), 
we work as follows: we flip the labels in all positions 
of $g$, i.e., 
$g_{00}, g_{01}, g_{10}, g_{11}$; if $g$ is the first 
input of $h$, we swap positions of $h$ horizontally: 
swap $h_{00}$ with $h_{10}$ and  $h_{01}$ with $h_{11}$; 
if $g$ is the second input of $h$, we perform 
a vertical swap: $h_{00}$ with $h_{01}$ and 
$h_{10}$ with $h_{11}$. This operation 
changes the state of $g$ from incorrect to correct
(or vice versa) and may also change correctness 
of $h$. If $g$ happens to be the $k$-th
output bit of circuit $j$, then this operation 
also flips position $c_k^{(j)}$.
See Figure~\ref{fig:circuit-reduction} for 
an example of two consecutive gates being flipped.
We call this permutation $\pi_g$. If $g$ is 
the $i$-th gate in circuit $j$, we may also call it 
$\pi_i^j$.
Note that 
if $\y$ is a well-behaved then $\y \circ \pi_g$ 
is well-behaved, too.\\

\textbf{Swapping two circuits.} We want a permutation
that simulates flipping the $i$-th input bit to 
$C$, the circuit in the instance of \FLIP. 
We achieve this by swapping $C_0$ 
with $C_i$---that is, swapping every position (input 
values, values in gate gadgets, output values) in 
$C_0$ with its corresponding position in $C_i$, 
and simultaneously flipping the $i$-th input 
bit $x_i^{(j)}$ for every circuit $C_j$ with $j \in \{1,\dots,n\} \setminus \{i\}$; naturally, 
if this $x_i^{(j)}$ is the first input to a gate 
$g$ in $C_j$, we have to perform the ``horizontal swap''
at $g$ 
outlined above, and if it is the second input to $g$,
a ``vertical swap'' at $g$. 
We call this permutation $\sigma_i$. 
Again, if $\y$ is well-behaved then $\y \circ \sigma_i$ 
is well-behaved, too. \\

\textbf{The starting string $\y_{\rm start} \in \{0,1\}^V$.}
This is the assignment where $C_0$ has input 
$\mathbf{0} \in \{0,1\}^n$ and $C_i$ has input 
$\e_i \in \{0,1\}^n$ and all gates have output $0$ (whether
correctly or incorrectly). This is certainly well-behaved 
(or rather, can be made well-behaved by making sure 
that input to gate $h$ matches output of gate $g$ 
should they be connected).\\ 

We now have a set $V$ of positions, an assignment 
$\y_{\rm start} \in \{0,1\}^V$, and  a
set of permutations--gate-flippers
$\pi_g$ for each gate $g$ 
and circuit-swappers $\sigma_i$ for each $i \in [n]$. 
They generate a subgroup $G$ of $S_V$. It is 
clear that the orbit of $\y_{\rm start}$ under $G$
is the set of well-behaved assignments.

\subsection{The cost function}

We have promised to use a lexicographic ordering 
as a cost function. That is, if $\y, \y' \in \{0,1\}^V$ are 
two assignments, then $\y$ is better than $\y'$ (meaning 
lower cost) if $\y \prec_{\rm lex} \y'$. Thus, 
to define the cost function it suffices to specify 
an ordering on the positions in $V$. 
\begin{itemize}
\item Positions in $C_0$ come before positions in $C_1$ 
and so on.
\item Within a circuit, most important are 
the control positions $g_{11}$ of the gates, followed
by the output gates, followed by all remaining positions.
\item Within control positions in the same circuit, 
the order follows the topological ordering of the circuit, i.e., if $g$ feeds into $h$, then $g$'s control position
comes before $h$'s.
\end{itemize}

\begin{figure}
    \centering
    \begin{tikzpicture}
    % Draw the outer square
    \draw[thick] (0,0) rectangle (4,4);
    
    % Draw the horizontal and vertical lines to divide into subsquares
    \draw[thick] (0,2) -- (4,2);
    \draw[thick] (2,0) -- (2,4);

    % Add text to each subsquare
    \node[align=left] at (1,3) {$g_{00}$ : 0\\ \textcolor{lightgray}{$g_{00}$ : 01}};
    \node[align=left] at (3,3) {$g_{01}$ : 1\\ \textcolor{lightgray}{$g_{01}$ : 10}};
    \node[align=left] at (1,1) {$g_{10}$ : 1\\ \textcolor{lightgray}{$g_{10}$ : 10}};
    \node[align=left] at (3,1) {$g_{11}$ : 1\\ \textcolor{lightgray}{$g_{11}$ : 10}};
    \node at (-0.5,3) {$y=0$};
    \node at (-0.5,1) {$y=1$};
    \node at (1, 4.5) {$x=0$};
    \node at (3, 4.5) {$x=1$};
    \end{tikzpicture}
    \caption{An example gate configuration with the input $x=0$ and $y=0$ and the output $0$. This is incorrect for a NAND-gate as indicated by the control bit. In light grey we denote the expanded encoding of the positions}
    \label{fig:enter-label}
\end{figure}

\begin{figure}[h!]
    \centering

    % Step 1
    \begin{subfigure}[t]{\textwidth}
        \centering
        \begin{minipage}[t]{0.35\textwidth}
            \centering
            \begin{tikzpicture}
                % Circuit style
                \ctikzset{
                    logic ports=european,
                    logic ports/scale=0.8,
                    tripoles/european not symbol=ieee circle,
                }
                 
                % Logic ports
                \node (x) at (-2, -0.18) {x};
                \node (y) at (-2, 0.185) {y};
                \node (z) at (-2, 1) {z};

                \node[nand port, red] (g1) at (0,0){};
                \node[nand port] (g2) at (2,0.5){};
                \node (output) at (3, 0.5) {};
                \node[above = 0 cm of g1] (g1text) {$g_1$};
                \node[above = 0 cm of g2] (g2text) {$g_2$};

                \draw (x) -- node [midway, below] {0} (g1.in 2);
                \draw (y) -- node [midway, above] {1} (g1.in 1);
                \draw (z) -| node [pos =0.2, above] {1} (g2.in 1);
                \draw (g1.out) -| node [pos=0.2, above] {\textit{0}} (g2.in 2) ;
                \draw (g2.out) -- node[pos=0.2, above] {\textit{1}} (output);
            \end{tikzpicture}
        \end{minipage}%
        \hfill
        \begin{minipage}[t]{0.55\textwidth}
            \centering
            \begin{tikzpicture}
                \draw[thick] (0,0) rectangle (.5,1.5);
                \draw[thick] (0,.5) -- (.5,.5);
                \draw[thick] (0,1) -- (.5,1);
                \node (x) at (.25, .25) {x:0};
                \node (y) at (.25, .75) {y:1};
                \node (z) at (.25, 1.25) {z:1};
                \node[above = of x] (ni) {Input};

                \node (ng1) at (1.5,1.25) {$g_1$};
                \draw[thick] (1,0) rectangle (2,1);
                \draw[thick] (1,.5) -- (2,.5);
                \draw[thick] (1.5,0) -- (1.5,1);
                \node (g10) at (1.25, .25) {\textit{0}};
                \node (g00) at (1.25, .75) {1};
                \node[red] (g11) at (1.75, .25) {1};
                \node (g01) at (1.75, .75) {1};
                
                \node (ng2) at (3.5,1.25) {$g_2$};
                \draw[thick] (3,0) rectangle (4,1);
                \draw[thick] (3,.5) -- (4,.5);
                \draw[thick] (3.5,0) -- (3.5,1);
                \node (g10) at (3.25, .25) {\textit{1}};
                \node (g00) at (3.25, .75) {0};
                \node (g11) at (3.75, .25) {0};
                \node (g01) at (3.75, .75) {0};

                \node[draw] (out) at (5, 0.25) {1};
                \node[above = 0cm of out] (no) {Output};

            \end{tikzpicture}
        \end{minipage}
        \subcaption{Step 1: Initial state. Gate $g_1$ is 
          evaluated incorrectly, $g_2$ correctly. Still, 
          we call this a {\em well-behaved} state.}
    \end{subfigure}

    \vspace{0.5cm} % Add space between rows

    \begin{subfigure}[t]{\textwidth}
        \centering
        \begin{minipage}[t]{0.35\textwidth}
            \centering
            \begin{tikzpicture}
                % Circuit style
                \ctikzset{
                    logic ports=european,
                    logic ports/scale=0.8,
                    tripoles/european not symbol=ieee circle,
                }
                 
                % Logic ports
                \node (x) at (-2, -0.18) {x};
                \node (y) at (-2, 0.185) {y};
                \node (z) at (-2, 1) {z};

                \node[nand port, red] (g1) at (0,0){};
                \node[nand port, red] (g2) at (2,0.5){};
                \node (output) at (3, 0.5) {};
                \node[above = 0 cm of g1] (g1text) {$g_1$};
                \node[above = 0 cm of g2] (g2text) {$g_2$};

                \draw (x) -- node [midway, below] {0} (g1.in 2);
                \draw (y) -- node [midway, above] {1} (g1.in 1);
                \draw (z) -| node [pos =0.2, above] {1} (g2.in 1);
                \draw (g1.out) -| node [pos=0.2, above] {\textit{0}} (g2.in 2) ;
                \draw (g2.out) -- node[pos=0.2, above] {\textit{0}} (output);
            \end{tikzpicture}
        \end{minipage}%
        \hfill
        \begin{minipage}[t]{0.55\textwidth}
            \centering
            \begin{tikzpicture}
                \draw[thick] (0,0) rectangle (.5,1.5);
                \draw[thick] (0,.5) -- (.5,.5);
                \draw[thick] (0,1) -- (.5,1);
                \node (x) at (.25, .25) {x:0};
                \node (y) at (.25, .75) {y:1};
                \node (z) at (.25, 1.25) {z:1};
                \node[above = of x] (ni) {Input};

                \node (ng1) at (1.5,1.25) {$g_1$};
                \draw[thick] (1,0) rectangle (2,1);
                \draw[thick] (1,.5) -- (2,.5);
                \draw[thick] (1.5,0) -- (1.5,1);
                \node (g10) at (1.25, .25) {\textit{0}};
                \node (g00) at (1.25, .75) {1};
                \node[red] (g11) at (1.75, .25) {1};
                \node (g01) at (1.75, .75) {1};
                
                \node (ng2) at (3.5,1.25) {$g_2$};
                \draw[thick] (3,0) rectangle (4,1);
                \draw[thick] (3,.5) -- (4,.5);
                \draw[thick] (3.5,0) -- (3.5,1);
                \node (g10) at (3.25, .25) {\textit{0}};
                \node (g00) at (3.25, .75) {1};
                \node (g11) at (3.75, .25) {\textcolor{red}{1}};
                \node (g01) at (3.75, .75) {1};

                \node[draw] (out) at (5, 0.25) {0};
                \node[above = 0cm of out] (no) {Output};

            \end{tikzpicture}
        \end{minipage}
        \subcaption{Step 2: Application of the permutation $\pi_{g_2}$ which leads to $g_2$ being now incorrect as well. This is not a local improvement step. While it minimizes the output, the more expensive control bit of $g_2$ is now active. Still the relationship between the gate state and the string under the current permutation holds, so 
        this is well-behaved as well.}

    \end{subfigure}
    \vspace{0.5cm}

    \begin{subfigure}[t]{\textwidth}
        \centering
        \begin{minipage}[t]{0.35\textwidth}
            \centering
            \begin{tikzpicture}
                % Circuit style
                \ctikzset{
                    logic ports=european,
                    logic ports/scale=0.8,
                    tripoles/european not symbol=ieee circle,
                }
                 
                % Logic ports
                \node (x) at (-2, -0.18) {x};
                \node (y) at (-2, 0.185) {y};
                \node (z) at (-2, 1) {z};

                \node[nand port] (g1) at (0,0){};
                \node[nand port] (g2) at (2,0.5){};
                \node (output) at (3, 0.5) {};
                \node[above = 0 cm of g1] (g1text) {$g_1$};
                \node[above = 0 cm of g2] (g2text) {$g_2$};

                \draw (x) -- node [midway, below] {0} (g1.in 2);
                \draw (y) -- node [midway, above] {1} (g1.in 1);
                \draw (z) -| node [pos =0.2, above] {1} (g2.in 1);
                \draw (g1.out) -| node [pos=0.2, above] {\textit{1}} (g2.in 2) ;
                \draw (g2.out) -- node[pos=0.2, above] {\textit{0}} (output);
            \end{tikzpicture}
        \end{minipage}%
        \hfill
        \begin{minipage}[t]{0.55\textwidth}
            \centering
            \begin{tikzpicture}
                \draw[thick] (0,0) rectangle (.5,1.5);
                \draw[thick] (0,.5) -- (.5,.5);
                \draw[thick] (0,1) -- (.5,1);
                \node (x) at (.25, .25) {x:0};
                \node (y) at (.25, .75) {y:1};
                \node (z) at (.25, 1.25) {z:1};
                \node[above = of x] (ni) {Input};

                \node (ng1) at (1.5,1.25) {$g_1$};
                \draw[thick] (1,0) rectangle (2,1);
                \draw[thick] (1,.5) -- (2,.5);
                \draw[thick] (1.5,0) -- (1.5,1);
                \node (g10) at (1.25, .25) {\textit{1}};
                \node (g00) at (1.25, .75) {0};
                \node (g11) at (1.75, .25) {0};
                \node (g01) at (1.75, .75) {0};
                
                \node (ng2) at (3.5,1.25) {$g_2$};
                \draw[thick] (3,0) rectangle (4,1);
                \draw[thick] (3,.5) -- (4,.5);
                \draw[thick] (3.5,0) -- (3.5,1);
                \node (g10) at (3.25, .25) {1};
                \node (g00) at (3.25, .75) {1};
                \node (g11) at (3.75, .25) {\textit{0}};
                \node (g01) at (3.75, .75) {1};

                \node[draw] (out) at (5, 0.25) {0};
                \node[above = 0cm of out] (no) {Output};

            \end{tikzpicture}
        \end{minipage}
        \subcaption{Step 3: Application of the permutation $\pi_{g_1}$. Due to the inversion, $g_1$ is now correct and the change of its output is reflected in $g_2$ which is now also correct since the one was swapped out of the control position. This is now a local optima.}

    \end{subfigure}
    \vspace{0.5cm}

    \caption{Step-by-step evaluation of the circuit and reduction process. Each step shows the circuit state (left) and the reduction state (right). Currently incorrect gates are marked red and the output of a gate is identifiable via italics.}
    \label{fig:circuit-reduction}
\end{figure}

\subsection{Proof of Correctness}

In this section we now prove the correctness and tightness of our reduction. In order to show the correctness we have to show that if $\pi \in G$ is a local optimum 
of the \localmin{$k$} instance, i.e., if 
the assignment $\y := \y_{\rm start} \circ \pi$ cannot 
be further improved by applying a $\pi_g$ or a $\sigma_j$, 
then the input to $C_0$ encoded 
in $\y$  corresponds to a local optimum of 
\FLIP.

Well-behaved permutations for a string $s$ and a circuit $C$ are interesting, because they realize the evaluation of the circuit somehow. The only problem is that the gate does not have the correct output in the current permutation with the string $s$ compared to $C$.

\begin{lemma}
    Let $\y$ be a well-behaved assignment and 
    suppose that $y(g_{11})=0$ for every gate---every 
    control position is labeled $0$. Then all output positions are mapped to correct results according to $C$.
    That is, if $\x \in \{0,1\}^n$ is the vector 
    to which the positions of 
    circuit $C_0$ are mapped under $\y$, then 
    \begin{enumerate}
    \item the input positions of $C_i$ are mapped to 
    $\x \oplus \e_i$ (this actually  holds for every 
    well-behaved $\y$, control positions being $0$ or $1$), 
    \item the output positions of $C_0$
    are mapped to $C(\x)$;
    \item the output positions of $C_j$ are mapped
     to $C(\x \oplus \e_j)$
    \end{enumerate}
\end{lemma}
\begin{proof}
    This follows from Observation~\ref{observation-control-bit}
    and induction over the sequence of gates. 
\end{proof}

The previous lemma tells us that all control positions should be mapped to zero in any local optimum so that all output positions are mapped to the correct output of $C$ given the input. We show that we can always apply a permutation to achieve this.

\begin{lemma}
    Let $\pi$ be a permutation from $G$. We consider a gate $g_i$ in the circuit $j$. The control position of $g_i$ is mapped to position of the form $g_{i,k}^j$ by $\pi$ 
\end{lemma}
\begin{proof}
This can be proven by induction on the structure of $\pi$ in the permutation group. Any generator preserves this property as it either does not affect this position at all $\pi^{j'}_{i'}$ for $i\ne i'$ or $j \ne j'$. If both $i'$ and $j'$ are equal to $i$ and $j$, then the position is simply inverted which preserves this property. Additionally, any permutation $\sigma_j$ either maps a gate to itself or swaps the gate, so that the claim holds here as well.
\end{proof}

\begin{lemma}
    If $\y$ is well-behaved and some 
    control position $g_{11}$ is $1$ under $\y$, then 
    $\y$ is not a local optimum.     
\end{lemma}
\begin{proof}
  Among all control positions assigned $1$, let $g_{11}$ 
  be the highest-ranking (i.e., of smallest index). 
  Now apply $\pi_g$, i.e., flip gate $g$, which inverts the bit by the previous lemma. Under 
  $\y \circ \pi_g$, the control position 
  $g_{11}$ is now correct; successor gates $h$ in the same circuit might now become incorrect, but their 
  control positions have lower rank by our ordering; 
  gates in other circuits are not affected. 
  Thus, $\y \circ \pi_g \prec_{\rm lex} \y$, and 
  $\y$ is not a local optimum.
\end{proof}

\begin{corollary}
    In any local optimum all sub circuits are correctly evaluated.
\end{corollary}

The previous lemmas suffice to show the main result.

\begin{lemma}
    Let $\y  = \y_{\rm start} \circ \pi \in \{0,1\}^V$ 
    be a local optimum an instance of 
    
    \noindent$\localmin{k}$. 
    Suppose the input variables of $C_0$ are mapped 
    to some $\x \in \{0,1\}^n$. Then $\x$ is a local optimum of the FLIP instance.
\end{lemma}
\begin{proof}
    Since $\y$ is well-behaved, the input variables 
    of $C_i$ are mapped to $\x \oplus \e_i$. 
    Since $\y$ is a local minimum, by 
    the corollary, the output values of $C_j$ are 
    in fact mapped to the correct output $C(\x \oplus \e_j)$. 

    Suppose now that the mapped input is not a local optimum for the \FLIP{} instance. Then there must be neighbor with a better output, i.e., 
    $C(\x \oplus e_j) \prec_{\rm lex} C(\x)$. 
    Now consider $\y' = \y \circ \sigma_j$. Under 
    $\y'$, the control positions of $C_0$ 
    are all mapped to $0$ (because those of $C_j$ were 
    under $\y$); the output of $C_0$ under $\y'$ is 
    better than under $\y$ because 
    $C(\x \oplus e_j) \prec_{\rm lex} C(\x)$. 
    Control positions in $C_i$ with $i \geq 1$ 
    might now be $1$, but their priority is less 
    than that of $C_0$'s output. Thus, $\y'$ is 
    better than $\y$, and $\y$ is not a local optimum.
\end{proof}

This concludes the proof of 
Theorem~\ref{theorem-PLS-complete}.
Every permutation used in the reduction has an order of two, so it is an involution. Still, these permutations do not form a commutative group due to the $\sigma_i$ permutations. This is to be expected as even finding a global optimum for the lexicographical leader of a string under an abelian permutation group where every element has an order of two is polynomial time solvable\cite[Section 3.1]{DBLP:conf/stoc/BabaiL83}.

We additionally note that the reduction is tight and hence finding a local optimal bitstring under permutations via the standard algorithm is \PSPACE-complete.

\begin{theorem}
    The given reduction is tight.
\end{theorem}
\begin{proof}
    Let $I$ be an instance of \FLIP\ and $(f,g)$ the previously defined reduction.
    We use as $\mathcal{R}$ simply the set of all permutations in the group of the generators which necessarily contains all local optima. We can find for any solution $s$ of $I$ in polynomial time a solution $\pi$ with $g(I,\pi) = s$ in $\mathcal{R}$ by applying the $\sigma_j$ permutations to construct the needed input string. This requires applying at most $n$ permutations which can be done in polynomial time. Finally, we see that any path where only the endpoints are in $\mathcal{R}$ must be an edge. If the edge is due to an $\pi_i^j$ permutation the input does not change and hence both endpoints are mapped to the same solution. If the edge is alternatively due to a $\sigma_j$ permutation this changes the input in one variable and is hence an edge in $TG(I)$ as it corresponds directly to a flip there. 
\end{proof}

\section{Realizing the permutations in propositional formula}

In the problem \localminsolution{} we are given a CNF formula $F$ over some variables $V$, a satisfying assignment 
$\alpha: V \rightarrow \{0,1\}$, and a list of permutations $\pi_1, \dots, \pi_k$ on $V$ such that 
$F$ is invariant under $\pi_i$ (that is, when applying $\pi_i$ to each variable occurrence in $F$, the resulting formula $F'$  is equal 
to $F$ up to a re-ordering of the clauses and the literals therein). The task now is to find a satisfying assignment $\beta$ of $F$ that 
such that $\beta \circ \pi_i \succeq_{\rm lex} \beta$.
This is clearly in PLS: whether $F$ is invariant under the $\pi_i$ and whether $\beta$ satisfies $F$ are both  easy to check. Note that it is 
not required that $\beta$ be in  the orbit of $\alpha$ under $\langle\pi_1,\dots,\pi_k\rangle$. 

\begin{theorem}
\localminsolution{} is PLS-complete.
\end{theorem}

\begin{proof}
We will define a formula $F$ whose satisfying assignments correspond exactly the 
well-behaved assignments to the positions, as defined in Definition~\ref{definition-well-behaved}. 
We first describe how to encode one circuit of the total $n+1$ circuits. We have input variables 
$x_1, \dots, x_n$ to the circuit; we introduce one {\em gate output variable} $g_{\rm out}$ for each gate; 
and four {\em gate control variables}
$g_{00}, g_{01}, g_{10}, g_{11}$ for the four positions in the square-representation of that gate, 
as in Figure~\ref{fig:enter-label}. For each such variable $u$ we introduce its twin $\tilde{u}$ and add 
$(u \leftrightarrow \neg \tilde{u})$. In other words, the (positive) literal $\tilde{u}$ simulates the negative literal $\bar{u}$. 
Take a gate $g$, let $u, v$ be its inputs and 
$w$ its output. 
The following formula $F_g$ ensures that the gate control variables $g_{ab}$ are set correctly as required 
for a well-behaved assignment:
\begin{align*}
(u \wedge v \wedge w) \rightarrow (\tilde{g}_{00} \wedge \tilde{g}_{01} \wedge \tilde{g}_{10} \wedge g_{11}) \\
(u \wedge v \wedge \tilde{w}) \rightarrow (g_{00} \wedge g_{01} \wedge g_{10} \wedge \tilde{g}_{11}) \\
(u \wedge \tilde{v} \wedge w) \rightarrow (\tilde{g}_{00} \wedge \tilde{g}_{01} \wedge g_{10} \wedge \tilde{g}_{11}) \\
(u \wedge \tilde{v} \wedge \tilde{w}) \rightarrow (g_{00} \wedge g_{01} \wedge \tilde{g}_{10} \wedge g_{11}) \\
(\tilde{u} \wedge v \wedge w) \rightarrow (\tilde{g}_{00} \wedge g_{01} \wedge \tilde{g}_{10} \wedge \tilde{g}_{11}) \\
(\tilde{u} \wedge v \wedge \tilde{w}) \rightarrow (g_{00} \wedge \tilde{g}_{01} \wedge g_{10} \wedge g_{11}) \\
(\tilde{u} \wedge \tilde{v} \wedge w) \rightarrow (g_{00} \wedge \tilde{g}_{01} \wedge \tilde{g}_{10} \wedge \tilde{g}_{11}) \\
(\tilde{u} \wedge \tilde{v} \wedge \tilde{w}) \rightarrow (\tilde{g}_{00} \wedge g_{01} \wedge g_{10} \wedge g_{11}) \\
\end{align*}
$F_g$ can easily be written as a CNF formula.
The permutation $\pi_g$ amounts to flipping the output of gate $g$ and flipping the corresponding inputs at those gates $h$ 
that $g$'s output feeds into. Thus, $\pi_g$ is 
\begin{align}
(w\tilde{w}) (g_{00} \tilde{g}_{00}) (g_{01} \tilde{g}_{01}) (g_{10} \tilde{g}_{10}) (g_{11} \tilde{g}_{11}) 
\circ (\textnormal{stuff at successor gates})
\label{permutation-flip-output}
\end{align}
$F_g$ is invariant under $\pi_g$ (even when we write it as a CNF). Next, 
there is a permutation $\sigma$ that flips the input $u$ of $g$. This swaps the two rows of the square representation of $g$: 
\begin{align}
\textnormal{(stuff at predecessor gates}) \circ (u \tilde{u}) 
(g_{00} g_{10})(\tilde{g}_{00} \tilde{g}_{10}) (g_{01} g_{10})(\tilde{g}_{01} \tilde{g}_{11})
\label{permutation-flip-input}
\end{align}
If $u$ is the output of some gate $h$, then $\sigma$ is $\pi_h$ and ``stuff at predecessor gate'' is what we describe 
in (\ref{permutation-flip-output}); otherwise $u$ is an input variable $x_i$ and ``stuff at predecessor gate'' does
nothing---for now. 

This formula describes well-behaved assignments in one of the $n+1$ circuits $C_0, \dots, C_n$. We now create 
$n+1$ copies of this formula, introducing a fresh version of each variable, so 
the \nth{$j$} input variable $x_j$ becomes $x^{(i)}_j$ in $C_i$, and $g_{00}$ becomes $g^{(i)}_{00}$ and so on. We create a 
formula $H$ to ensure that $x^{(i)}_j$ and $x^{(0)}_j$ differ if and only if  $i = j$: 
\begin{align}
  H := \bigwedge_{\{i_1, i_2\} \in {\{0,\dots,n\} \choose 2}} \bigwedge_{j=1}^n 
  \left\{ 
  \begin{array}{cc}
    \left(x^{(i_1)}_j \leftrightarrow x^{(i_2)}_j\right) \wedge 
    \left(\tilde{x}^{(i_1)}_j \leftrightarrow \tilde{x}^{(i_2)}_j\right) 
    & \textnormal{ if $j \not \in \{i_1, i_2\}$} \\
    & \\
    \left(x^{(i_1)}_j \leftrightarrow \tilde{x}^{(i_2)}_j\right) \wedge 
    \left(\tilde{x}^{(i_1)}_j \leftrightarrow x^{(i_2)}_j\right)
    & \textnormal{ if $j \in \{i_1, i_2\}$}  \ . 
  \end{array}
  \right\}
  \label{formula-syntactically-correct}
\end{align}
The permutation $\sigma_i$ flipping circuit $C_i$ and $C_0$ and inverting $x^{(i')}_i$ for all other $i'$ can be written as 
\begin{align*}
& \left(u^{(i)}_j u^{(0)}_j \right)\left(\tilde{u}^{(i)}_j \tilde{u}^{(0)}_j \right) 
\tag{for each input and gate output and control variable $u$}
\\
\circ & 
\left( x^{(i')}_i \tilde{x}^{(i')}_i \right) \tag{for all $i' \in \{1,\dots,n\} \setminus \{i\}$} 
\circ \left(\textnormal{do (\ref{permutation-flip-input}) at gates having $x^{(i')}_i$ as input} \right)
\end{align*}
This forms the final formula $F = H \wedge \bigwedge_{i=0}^n \bigwedge_{\textnormal{ gate $g$}} F^{(i)}_g$. Its satisfying assignments 
correspond exactly to the well-behaved assignments described above
and that each $\pi_g$ and each $\sigma_i$ is indeed a symmetry of $F$. The order of the variables is as in the previous proof: of highest priority 
are the control variables $g^{(0)}_{11}$ (following the topological order of the gates in the circuit); 
then the output variables of $C_0$; then to the control variables of the other circuits; then all the rest.
It is easy to provide an ``initial'' satisfying assignment $\alpha \in \sat(F)$. 
Finally, if $\beta$ is some satisfying assignment of $F$ that is locally minimal, i.e., cannot be improved by applying any 
$\pi_g$ or $\sigma_i$, then $\beta$ represents a configuration in which all circuits $C_0, \dots, C_n$ are correctly evaluated and 
no $C_i$ outputs something better than $C_0$; in other words, a local optimum of FLIP. This shows that 
\localminsolution{} is PLS-complete.
\end{proof}

\section{Conclusion and open questions}

We have shown that the problem is \PLS-hard in general and hence a polynomial time algorithm is unlikely unless \P=\PLS. In the theory of \PLS-complete problems we thereby demonstrated another example of a hard problem with lexicographic weights. The used permutations are realistic in the sense that they can occur in an actual CNF formula. 

Our above reduction requires polynomially many permutations. There is an efficient algorithm 
for the case of {\em one} permutation. What about when we are given a constant number of  permutations?  

What about if the given permutations form an Abelian group? Is it still PLS-hard to find a local minimum?

% Additionally, we see that finding a local minimum of an Abelian permutation group is P-hard.
% Is this a strict lower bound, i.e. does there exist a polynomial time algorithm to compute a local optima in this case or can the reduction be modified to work also for Abelian permutation groups. Another possibility would be that this is in an intermediate class such as \CLS, but it seems unclear to us how to deal with multiple permutations that all have the same successor permutation.

The neighborhood of a solution $\pi$ is in our work defined as $\pi \circ \pi_i$ for some generator $\pi_i$. An alternative neighborhood would be $\pi_i \circ \pi$, which would place the generator between the string and the current permutation. While the results of the one permutation case trivially hold again, the hardness proof does not transform to this formulation.
\bibliography{bibliography}

\appendix

\section{Condensed and expanded view.}

We expand on the encoding of a position $x$ with two positions $x_0$ and $x_1$ which was needed to allow the number of ones and zeros to change.
The permutations make sure that $x_0$ and $x_1$ always stick together as inverting a position swaps $x_0$ and $x_1$ and swaps between circuits swap both elements. This allows us to always change a position by inverting it.
Additionally, we note that we can order the positions in a way such that any local optimum in the condensed view is a local optima is the expanded view and vice versa. Let $X$ be the set of positions and $f: X \to \mathbb{N}$ be a function mapping the positions to their importance. Then we can use the alternative map $f'$ with $f'(x_0) = 2 f(x)$ and $f'(x_1)=2 f(x) + 1$. 

In a local optima $y$ of the condensed view there exists for every permutation $\pi$ in the group a smallest index $i$ at which $y$ and $y \circ \pi$ differ. Then they differ also at $2i$ in the expanded view. All position before $2i$ do not differ in the expanded view as well and since $y$ is a local optima, the position $2i$ in the expanded view is mapped to zero. In the expanded view of $y \circ \pi$ the position $2i$ must be mapped to one by the assumption and hence it is not improving here as well. 

\end{document}